\documentclass[submission,copyright,creativecommons]{eptcs}
\usepackage[utf8]{inputenc}
\providecommand{\event}{HCVS 2021} 
\usepackage{underscore}            
\usepackage{amsmath}
\usepackage{amssymb}

\usepackage{amsthm}
\theoremstyle{plain}
\newtheorem{theorem}{Theorem}
\newtheorem{lemma}[theorem]{Lemma}
\newtheorem{corollary}[theorem]{Corollary}
\newtheorem{conjecture}[theorem]{Conjecture}
\theoremstyle{definition}
\newtheorem{definition}[theorem]{Definition}
\newtheorem{example}[theorem]{Example}

\newcommand{\calA}{\mathcal{A}}
\newcommand{\calB}{\mathcal{B}}
\newcommand{\calL}{\mathcal{L}}
\newcommand{\calM}{\mathcal{M}}
\newcommand{\compl}{^\mathrm{c}}
\newcommand{\mybar}[1]{\overline{#1}}
\newcommand{\Nat}{\mathbb{N}}
\newcommand{\Int}{\mathbb{Z}}
\newcommand{\Rat}{\mathbb{Q}}
\newcommand{\id}{{\mathrm{id}}}
\newcommand{\powset}{\mathcal{P}}
\newcommand{\union}{\cup}
\newcommand{\sop}{[} 
\newcommand{\scl}{]} 
\newcommand{\sel}[2]{#1 \backslash #2} 
\newcommand{\unsubst}[2]{\sop \sel{#1}{#2} \scl} 
\newcommand{\impl}{\rightarrow} 
\newcommand{\liff}{\leftrightarrow} 
\newcommand{\Land}{\bigwedge}
\newcommand{\Lor}{\bigvee}
\newcommand{\proves}{\vdash}
\newcommand{\lfp}{\mathrm{lfp}} 
\newcommand{\hoare}[3]{\{#1\}#2\{#3\}}
\newcommand{\vc}{\mathrm{vc}}
\newcommand{\vct}{\tilde{\vc}}
\newcommand{\hskipit}{\textbf{skip}}
\newcommand{\hif}{\textbf{if}~}
\newcommand{\hthen}{~\textbf{then}~}
\newcommand{\helse}{~\textbf{else}~}
\newcommand{\hwhile}{\textbf{while}~}
\newcommand{\hdo}{~\textbf{do}~}
\newcommand{\wpc}{\mathrm{wp}}
\newcommand{\spc}{\mathrm{sp}}
\newcommand{\Laff}{\calL_{\mathrm{aff}}}
\newcommand{\Aff}{{\mathrm{Aff}~}}
\newcommand{\aff}{{\mathrm{aff}}}
\newcommand{\modelsa}{\models_\mathrm{a}} 

\title{A Fixed-point Theorem for Horn Formula Equations}
\author{Stefan Hetzl \qquad\qquad Johannes Kloibhofer\thanks{Supported by the Vienna Science
 and Technology Fund (WWTF) as part of the Vienna Research Group VRG 12-004}
\institute{Institute of Discrete Mathematics and Geometry}
\institute{TU Wien\\
Vienna, Austria}
\email{stefan.hetzl@tuwien.ac.at \qquad\quad j.kloibhofer@gmx.net}
}
\def\titlerunning{A Fixed-point Theorem for Horn Formula Equations}
\def\authorrunning{S.~Hetzl \& J.~Kloibhofer}

\begin{document}
\maketitle

\begin{abstract}
We consider constrained Horn clause solving from the more general point of view
 of solving formula equations.
Constrained Horn clauses correspond to the subclass of Horn formula equations.
We state and prove a fixed-point theorem for Horn formula equations which is based on expressing
 the fixed-point computation of a minimal model of a set of Horn clauses on the object level
 as a formula in first-order logic with a least fixed point operator.
We describe several corollaries of this fixed-point theorem, in particular concerning the logical
 foundations of program verification, and sketch how to generalise
 it to incorporate abstract interpretations.
\end{abstract}

%

\section{Introduction}

Constrained Horn clauses are a versatile and practical formalism for representing and solving
 a variety of problems in program verification and model checking~\cite{Bjorner15Horn,Gurfinkel19Science}.
In this paper we approach constrained Horn clause solving from a theoretical point of view.
In logic, related problems have a long history:
solving Boolean equations goes back to the 19th century and was already investigated
 in \cite{Schroeder90Vorlesungen}, see~\cite{Rudeanu74Boolean} for a comprehensive textbook.
Solving Boolean equations is closely related to Boolean unification, a subject of thorough study in
 computer science, see, e.g., \cite{Martin89Boolean} for a survey.
The generalisation of this problem from propositional to first-order logic has been made explicit
 as early as~\cite{Behmann50Aufloesungsproblem,Behmann51Aufloesungsproblem}.
Solving a formula equation in first-order logic is closely related to second-order quantifier elimination,
 a problem with applications in a variety of areas in computer science, e.g., databases or common-sense
 reasoning~\cite{Gabbay08Second}.
A seminal work on second-order quantifier elimination and basis of several algorithms still in use today
 is Ackermann's~\cite{Ackermann35Untersuchungen}.
See~\cite{Wernhard17Boolean} for a recent survey of this area of related problems.

Solving constrained Horn clauses is closely related to solving formula equations.
In fact, constrained Horn clauses correspond to a natural class of formula
 equations which we will call Horn formula equations.
This relationship allows for an elegant theoretical description of the connections between
 problems considered from Ackermann to contemporary verification.

Fixed-point theorems play an important role for solving equations in many areas of mathematics.
In recursion theory, a subject with close ties to verification, the recursion theorem guarantees the
 existence of a solution of a system of recursion equations by computing a fixed point.
But also in areas quite remote from verification similar constructions can be found, as in the use of Banach's
 fixed-point theorem in the proof of the
 Picard-Lindelöf theorem on the unique solvability of ordinary differential equations.
In constrained Horn clause solving we have a comparable situation: it is a well-known result from
 logic programming and constraint logic programming that every set of Horn clauses has a unique minimal model
 (in the sense of identifying a model with the ground atoms true in it) and that this minimal model
 can be computed as the fixed-point of an operator induced by the clause set~\cite{VanEmden76Semantics,Jaffar94Constraint}.
 
In this paper we formulate and prove a fixed-point theorem for Horn formula equations which (essentially)
 makes the construction of the minimal model explicit {\em in the logic}.
Expressing this construction will be achieved by using FO[LFP], first-order logic with a least fixed-point operator,
 thus providing a canonical solution for any Horn formula equation.
The fixed-point theorem has a number of applications:
it helps to explain, at least from a theoretical  point of view, the efficacy of interpolation for Horn
 clause solving and invariant generation.
Moreover, as a simple corollary one can obtain the expressibility of the 
 weakest precondition and the strongest postcondition, and thus the partial correctness of an imperative
 program in FO[LFP].
As another corollary it allows a generalisation of a result by
 Ackermann~\cite{Ackermann35Untersuchungen} on second-order quantifier-elimination in a direction different
 from the recent generalisation~\cite{Wernhard17Approximating}.
A result from a recently introduced approach to inductive theorem proving with tree grammars described
 in~\cite{Eberhard15Inductive} on generating a proof with induction based on partial information about
 that proof can be obtained from our fixed-point theorem as another straightforward corollary.
Last, but not least, an abstract form of the fixed-point theorem, stated here as Conjecture~\ref{conj.afpthm}, would
 permit to considerably simplify the proof of the decidability of affine formula
 equations~\cite{Hetzl20Decidability}.

In Section~\ref{sec.relwork} we relate constrained Horn clause solving with solving formula equations.
The fixed-point theorem is stated and proved in detail in Section~\ref{sec.fpthm}.
Section~\ref{sec.apppv} describes some of its applications to the foundations of program verification.
In Section~\ref{sec.afpthm} we sketch how to generalise our fixed-point theorem to accommodate
 abstract interpretation based on Galois connections.
This paper is an improved presentation of some of the main results of the second authors master's
 thesis~\cite{Kloibhofer20Fixed}.

\section{Constrained Horn clauses and formula equations}\label{sec.relwork}

We use standard notation from logic.
In particular, for a first-order language $\calL$, an $\calL$-structure $\calM$, an $\calL$-formula $\varphi$,
 and an interpretation $I$ of the free variables of $\varphi$ in $\calM$ we write $\calM,I \models \varphi$ to
 express that $\varphi$ is true in $\calM$ under the interpretation $I$ (in the usual sense of Tarski semantics).
Occasionally we will, in a slight abuse of notation, also allow elements of $\calM$ or relations over $\calM$ to
 appear on the right-hand side of $\models$ by which we intend to denote truth in $\calM$ under
 an accordingly modified interpretation $I$.
We write $\models \varphi$ to express that $\varphi$ is a valid formula.
Individual variables will be denoted by lowercase Latin letters $x,y,z,u,v\ldots$. 
Predicate variables will be denoted by uppercase Latin letters $X,Y,\ldots$.
If $X$ is a $k$-ary predicate variable and $\chi$ is a first-order formula with the free variables
 $v_1,\ldots,v_k$ we write $\unsubst{X}{\chi}$ for the substitution of $X$ by $\chi$ inserting the $i$-th
 argument of an $X$-atom for $v_i$.
We write $\sop \sel{X_1}{\chi_1}, \ldots, \sel{X_n}{\chi_n} \scl$ for the
 simultaneous substitution of $X_i$ by $\chi_i$ for $i=1,\ldots,n$.
A substitution $\sop \sel{X_1}{\chi_1},\ldots,\sel{X_n}{\chi_n} \scl$ is
 called {\em first-order substitution} if $\chi_1,\ldots,\chi_n$ are first-order formulas.
The logical symbol $\bot$ is a nullary predicate constant which is false in all structures.

Let $\calL$ be a first-order language and let $T$ be an $\calL$-theory.
A {\em constrained Horn clause} is an $\calL$-formula of the form
 $\varphi \land X_1(\mybar{t_1}) \land \cdots \land X_n(\mybar{t_n}) \impl Y(\mybar{s})$ or
 $\varphi \land X_1(\mybar{t_1}) \land \cdots \land X_n(\mybar{t_n}) \impl \bot$
 where $\mybar{t_1},\ldots,\mybar{t_n},\mybar{s}$ are
 tuples of first-order terms of appropriate arity and $\varphi$ is a first-order formula, i.e., a formula not containing
 predicate variables.
Note that a constrained Horn clause is allowed to (and typically does) contain free individual variables which,
 as usual in clause logic, are treated as universally quantified.
A finite set $S$ of constrained Horn clauses is considered as the conjunction of these clauses and is thus
 logically equivalent to a formula of the form $\forall^* \Land_{C\in S} C$ where $\forall^*$ denotes
 the universal closure w.r.t.\ individual variables.
We are interested in solving a given finite set of constrained Horn clauses.
There are different notions of solvability in the literature which we will discuss in detail below.
 
In this paper we embed constrained Horn clauses in the more general framework of formula equations.
In the context of logical formulas, we consider an equation to be 
\begin{equation}\label{eq.feqtwosided}
\text{an $\calL$-formula}\ \varphi_1 \liff \varphi_2\ \text{containing predicate variables $\mybar{X} = X_1,\ldots,X_n$.}
\end{equation}
A solution of~(\ref{eq.feqtwosided}) is a first-order substitution $\unsubst{\mybar{X}}{\mybar{\chi}}$
 s.t.\ $\models (\varphi_1 \liff \varphi_2) \unsubst{\mybar{X}}{\mybar{\chi}}$.
A solution in $T$ has to satisfy $T \models (\varphi_1 \liff \varphi_2)\unsubst{\mybar{X}}{\mybar{\chi}}$ instead.
Since the concept of solution in $T$ subsumes that of solution (by setting $T = \emptyset$) we will only use the
 former.
We can simplify equations to instead considering
\begin{equation}\label{eq.feqonesided}
\text{an $\calL$-formula}\ \varphi\ \text{containing predicate variables $\mybar{X} = X_1,\ldots,X_n$.}
\end{equation}
Then, as a solution in $T$, we ask for a first-order substitution $\unsubst{\mybar{X}}{\mybar{\chi}}$ s.t.\ 
 $T \models \varphi\unsubst{\mybar{X}}{\mybar{\chi}}$.
Note that every instance of~(\ref{eq.feqtwosided}) is an instance of~(\ref{eq.feqonesided}) by letting $\varphi$ be $\varphi_1 \liff \varphi_2$
 and every instance of~(\ref{eq.feqonesided}) is an instance of~(\ref{eq.feqtwosided}) by letting $\varphi_1$ be $\varphi$ and 
 $\varphi_2$ be $\top$.
Moreover, it will be notationally useful to explicitly indicate the predicate variables by 
 existential quantifiers.
Consequently we define:
\begin{definition}
A {\em formula equation} is a closed $\calL$-formula $\exists \mybar{X}\, \varphi$ where
 $\varphi$ contains only first-order quantifiers.
A {\em solution} of $\exists \mybar{X}\, \varphi$ in $T$ is a first-order substitution
 $\sop\sel{\mybar{X}}{\mybar{\chi}}\scl$ s.t.\  $T \models \varphi\unsubst{\mybar{X}}{\mybar{\chi}}$.
\end{definition}
The problem of computing a solution to a formula equation given as input will be denoted as FEQ in the sequel.
A formula equation $\exists \mybar{X}\, \varphi$ is called {\em valid} if it is a valid second-order formula
 and {\em satisfiable} if it is a satisfiable second-order formula.
If $S$ is a set of constrained Horn clauses in the predicate variables $\mybar{X} = X_1,\ldots,X_n$, then
 $\exists \mybar{X} \forall^* \Land_{C\in S} C$ will be called {\em Horn formula equation}.
Thus constrained Horn clauses correspond to existential second order Horn logic, which also plays a
 significant role in finite model theory, see~\cite{Graedel91Expressive}.

There are different notions of solvability for constrained Horn clauses in the literature:
{\em satisfiability} of~\cite{Gurfinkel19Science} is satisfiability of a Horn formula equation,
{\em semantic solvability} of~\cite{Ruemmer13Classifying} is validity of a Horn formula equation, and
{\em syntactic solvability} of~\cite{Ruemmer13Classifying} is solvability of a Horn formula equation.
In this paper we will primarily be interested in this last notion: solvability of a Horn formula equation.
Every solvable formula equation is valid and every valid formula equation is satisfiable but neither of
 the converse implications are true as the following example shows.
\begin{example}\label{ex.feq}
If $\varphi$ is a first-order formula which is satisfiable but not valid and does not contain $X$
 then, trivially, $\exists X\, \varphi$ is a formula equation which is satisfiable but not valid.
 
Towards an example for a valid but unsolvable Horn formula equation we work in the first-order language
 $\calL = \{ 0/0, s/1 \}$.
Let $A_1$ be $\forall x\, s(x)\neq 0$ and let $A_2$ be $\forall x\forall y\, (s(x)=s(y)\impl x=y)$
 and consider the formula
\[
A_1 \land A_2 \impl \exists X \exists Y \forall u\, \Big( X(0) \land Y(s(0)) \land (X(u) \impl Y(s(u)))
 \land (Y(u) \impl X(s(u))) \land \neg (X(u) \land Y(u)) \Big)
\]
which, up to some simple logical equivalence transformations, is a Horn formula equation $\Phi$.
Now $\Phi$ is valid since, in a model $\calM$ of $A_1 \land A_2$, interpreting $X$ by
 $\{ s^{2n}(0) \mid n \in \Nat \}$ and $Y$ by $\{ s^{2n+1}(0) \mid n \in \Nat \}$ makes the remaining
 formula true.

For unsolvability suppose that $\Phi$ has a solution $\sop \sel{X}{\chi(u)}, \sel{Y}{\psi(u)} \scl$,
 then, since the standard model $\Nat$ in the language $\calL$ satisfies $A_1\land A_2$, we would have
\[
 \Nat \models \chi(0)\land \psi(s(0)) \land \forall u\, (\chi(u) \impl \psi(s(u))) \land \forall u\, (\psi(u) \impl \chi(s(u))) \land \forall u\, \neg (\chi(u) \land \psi(u)),
\]
in particular $\chi$ would be a definition of the even numbers.
However, the theory of $\Nat$ in $\calL$ admits quantifier elimination~\cite[Theorem 31G]{Enderton01Mathematical}
 which has the consequence that the $L$-definable sets in $\Nat$ are the finite and co-finite
 subsets of $\Nat$~\cite[Section 3.1, Exercise 4]{Enderton01Mathematical} and thus we obtain a contradiction to $\chi$ being
 a definition of the even numbers (which is neither finite nor co-finite).
\end{example}

Solving formula equations (FEQ) is closely related to the problem of second-order quantifier elimination (SOQE):
given a formula $\exists \mybar{X}\, \varphi$ where $\varphi$ contains only first-order quantifiers find
 a first-order formula $\psi$ s.t.\ $\models \exists \mybar{X}\, \varphi \liff \psi$.
The relationship between FEQ and SOQE often hinges on a third problem: second-order quantifier elimination
 by a witness (WSOQE): given a formula $\exists \mybar{X}\, \varphi$ where $\varphi$ contains only first-order 
 quantifiers find a first-order substitution $\unsubst{\mybar{X}}{\mybar{\chi}}$
 s.t.\ $\models \exists \mybar{X}\, \varphi \liff \varphi\unsubst{\mybar{X}}{\mybar{\chi}}$,
 see~\cite{Wernhard17Boolean} for more details.
This complex of problems has a long history in logic and a wealth of applications in computer science,
see the textbook~\cite{Gabbay08Second} on second-order quantifier elimination.
A number of algorithms for second-order quantifier elimination have been developed, for example:
The SCAN algorithm introduced in~\cite{Gabbay92Quantifier} (tries to) compute(s) a first-order formula equivalent to $\exists X\, \varphi$
 for a conjunctive normal form $\varphi$ by forming the closure of $\varphi$ under constraint resolution which
 only resolves on $X$-literals.
The DLS algorithm has been introduced in~\cite{Doherty97Computing} and consists essentially of
 formula rewriting steps tailored to allow application of Ackermann's lemma (which instantiates a predicate
 variable provided some conditions on the polarity of its occurrences are met).

\section{The fixed-point theorem}\label{sec.fpthm}

It is well-known that a set of Horn clauses has a minimal model and that it can be obtained as
 least fixed point of an operator induced by the clause set.
In order to integrate this insight seamlessly into the framework of formula equations we will express
 it on the object level by representing this least fixed-point by means of an explicit
 least fixed-point operator.
An adequate tool to that end is first-order logic with least fixed points, FO[LFP], which plays an
 important role in finite model theory and descriptive complexity~\cite{Immerman99Descriptive}.

In order to introduce FO[LFP] we first define, by induction on a formula $\varphi$, what it means
 for a predicate variable $X$ to occur positively (negatively) in $\varphi$:
\begin{enumerate}
\item If $\varphi$ is an atom of the form $X(\mybar{t})$ then $X$ occurs positively in $\varphi$.
\item $X$ occurs positively (negatively) in $\neg \varphi$ iff $X$ occurs negatively (positively) in $\varphi$.
\item For $\circ \in \{ \land, \lor \}$, $X$ occurs positively (negatively) in $\varphi \circ \psi$ iff
 $X$ occurs positively (negatively) in $\varphi$ or $X$ occurs postively (negatively) in $\psi$.
\item $X$ occurs positively (negatively) in $\varphi \impl \psi$ iff $X$ occurs positively (negatively)
 in $\psi$ or $X$ occurs negatively (positively) in $\varphi$.
\item For $Q \in \{ \forall, \exists \}$, $X$ occurs positively (negatively) in $Q x\, \varphi$ iff
 $X$ occurs positively (negatively) in $\varphi$.
\end{enumerate}
\begin{example}
$X$ occurs positively but not negatively in $\forall u\, (P(u) \impl X(u))$.
$X$ occurs both positively and negatively in $\forall u\, (X(u) \impl X(s(u)))$.
\end{example}
FO[LFP] is first-order logic, augmented with a least fixed point operator $\lfp$ which allows
 to add new predicates to the logic that have the form $[\lfp_X\, \varphi(X,\mybar{u})]$ where
 $\varphi$ is a formula in which $X$ occurs only positively
 and the length of $\mybar{u}$ is the arity of $X$.
Then $\varphi$ defines a monotone function $F_\varphi: \powset(M)^k \to \powset(M)^k$ and since
 the power set lattice $\powset(M)^k$ is complete, the Knaster-Tarski theorem applies and the
 least fixed point of $F_\varphi$ is well-defined.
The predicate $[\lfp_X\, \varphi(X,\mybar{u})]$ is interpreted as that least fixed point.
\begin{example}
Working in the language $\{E/2\}$ of graphs, let $X$ be a binary predicate variable and define
\[
\varphi(X,u,v)\quad\equiv\quad u = v \lor \exists w (X(u,w) \wedge E(w,v)).
\]
As $X$ occurs only positively in $\varphi$ we can form $[\lfp_X\, \varphi(X,u,v)]$ and
 observe that $[\lfp_X\, \varphi(X,u,v)](a,b)$ is true iff there is a path from $a$ to $b$.
\end{example}
In this paper we will consider first-order logic with an operator for simultaneous least fixed points
 (which corresponds to introducing mutual recursion).
Then we require a tuple $\Phi = (\varphi_i(X_1,\ldots,X_n,\mybar{u_i}))_{i=1}^n$ of formulas containing the
 $X_i$ only positively where the length of $\mybar{u_i}$ is the arity $k_i$ of $X_i$.
For a structure $\calM$ with domain $M$ we define 
\begin{align*}
F_i:\quad  M^{k_1} \times \cdots \times M^{k_n} &\rightarrow M^{k_i},\\
(R_1,\ldots,R_n) &\mapsto \{\mybar{a} \in M^{k_i} \mid \calM \models \varphi_i(R_1,\ldots,R_n,\mybar{a}) \}.
\end{align*}
and the operator $F_{\Phi} = (F_1,...,F_n) : M^{k_1} \times \cdots \times M^{k_n} \to M^{k_1} \times \cdots \times M^{k_n}$.
Since the $X_i$ occur only positively in the $\varphi_j$, the operator $F_\Phi$ is monotone and, again,
 the Knaster-Tarski theorem applies.
Thus we obtain new predicates $[\lfp_{X_i} \Phi]$ for $i=1,\ldots,n$ which are interpreted by the $i$-th
 component of the least fixed point of the operator $F_\Phi$.
For more details, the reader is referred to~\cite{Dawar02Fixed}.

Let $\exists X_1 \cdots \exists X_n\, \psi$ be a Horn formula equation.
We distinguish three different types of clauses in $\psi$:
%
\begin{align*}
	\begin{array}{crl}
	(B) &\varphi &\impl X_{i_0}(\mybar{s}),   \\
	(I) &\varphi \land X_{i_1}(\mybar{t_1}) \land \cdots \land X_{i_m}(\mybar{t_m}) &\impl X_{i_0}(\mybar{s}),  \\
	(E) &\varphi \land X_{i_1}(\mybar{t_1}) \land \cdots \land X_{i_m}(\mybar{t_m}) &\impl \bot, 
	\end{array}
\end{align*}
where the constraint $\varphi$ is a formula in $\calL$ not containing a predicate variable,
 $m \geq 1$, $\mybar{t_1},..,\mybar{t_m},\mybar{s}$ are
 tuples of first-order terms in $\calL$ of appropriate arity and $i_0, i_1,\ldots,i_m \in \{1,\ldots,n\}$. 
Note that free variables $\mybar{y}$ may occur in the formulas $\varphi$ and the terms
 $\mybar{s}, \mybar{t_1},\ldots,\mybar{t_m}$.
We call the first base clauses, the second induction clauses, and the third end clauses.
The idea now is to build an inductive relation from the base and induction clauses for every
 formula variable.
For $j=1,\ldots,n$ let $B_j$ and $I_j$ be the sets of clauses of the form (B) and (I), respectively,
 where $i_0 = j$.
In order to abbreviate notation we write $ \iota := (i_1, \ldots ,i_m)$ and $\tau := (\mybar{t_1}, \ldots ,\mybar{t_m})$.
A clause in $I_j$ is determined by the tuple $(\varphi,\iota,\tau,\mybar{s})$, thus we write
 $(\varphi,\iota,\tau,\mybar{s})$ for the clause
 $\varphi \land X_{i_1}(\mybar{t_1}) \land \cdots \land X_{i_m}(\mybar{t_m}) \impl X_{j}(\mybar{s})$
 in $I_j$.
Analogously we write $(\varphi,\mybar{s})$ for the clause $\varphi \impl X_{j}(\mybar{s})$ in $B_j$.
\begin{definition}
Let $\exists X_1 \cdots \exists X_n\, \psi$ be a Horn formula equation.
Define the $n$-tuple $\Phi_\psi = (\varphi_j(X_1,\ldots,X_n,\mybar{x_j}))_{j=1}^n$ where, for $j=1,\ldots,n$,
\[
\varphi_j(X_1,\ldots,X_n,\mybar{x_j})
\quad \text{is}\quad
\exists \mybar{y} \left(
 \Lor_{(\varphi,\mybar{s}) \in B_j} \left(\varphi \land \mybar{x_j} = \mybar{s}\right) \lor
 \Lor_{(\varphi,\iota,\tau,\mybar{s}) \in I_j} \left(\varphi \land \Land_{k = 1}^{m}  X_{i_k}(\mybar{t_k})
 \land \mybar{x_j} = \mybar{s} \right)\right)
\]
where $\mybar{y}$ are the free variables of the clauses in $B_j \union I_j$ and $\mybar{x_j}$ is a
 tuple of variables s.t. $|\mybar{x_j}|$ equals the arity of $X_j$.
\end{definition}
From the point of view of (constraint) logic programming, the above tuple of formulas is a first-order definition
 of the operator $T_P$ induced by $\exists \mybar{X}\, \psi$ when considered as a constraint logic program $P$,
 see, e.g.,~\cite{Jaffar94Constraint}.
Note that $X_1,\ldots,X_n$ only occur positively in $\Phi_{\psi}$, hence we can introduce the simultaneous
 fixed-point formulas $[\lfp_{X_j} ~\Phi_{\psi}]$ for $j \in \{1,\ldots,n\}$.

\begin{lemma}\label{lem.HornFP}
Let $\exists X_1\cdots\exists X_n\, \psi$ be a Horn formula equation and $\mu_j := [\lfp_{X_j} ~\Phi_{\psi}]$
 for $j \in \{1,\ldots,n\}$, then
\begin{enumerate}
\item\label{lem.HornFP.sol} $\models \exists \mybar{X}\,\psi \liff \psi\unsubst{\mybar{X}}{\mybar{\mu}}$ and
\item\label{lem.HornFP.minsol} if  $\calM \models \psi \unsubst{\mybar{X}}{\mybar{R}}$
 for some structure $\calM$ and relations $R_1,\ldots,R_n$ in $\calM$, then
 $\calM \models \Land_{j=1}^n ( \mu_j \impl R_j )$.
\end{enumerate}
\end{lemma}
\begin{proof}
The right-to-left direction of~\ref{lem.HornFP.sol} is clear.
For the left-to-right direction we first observe that the formulas $\mu_1,\ldots,\mu_n$ satisfy all
 clauses in (B) and (I), i.e., for all $j\in \{1,\ldots,n\}$ we have
\begin{align*}
\models & \forall \mybar{y}\, (\varphi \impl \mu_j(\mybar{s})), &\forall (\varphi,\mybar{s}) \in B_j,\\
\models & \forall \mybar{y}\, (\varphi \land \mu_{i_1}(\mybar{t_1}) \land \cdots \land \mu_{i_m}(\mybar{t_m})
 \impl \mu_j(\mybar{s})), &\forall (\varphi,\iota,\tau,\mybar{s}) \in I_j.
\end{align*}
To see this let $\calM$ be a structure and $\mybar{a}$ s.t.\ 
$\calM,[\mybar{y} := \mybar{a}] \models \varphi \land \mu_{i_1}(\mybar{t_1}) \land \cdots \land \mu_{i_m}(\mybar{t_m})$,
then, as $(\mu_1,\ldots,\mu_n)$ is a fixed point of $F_{\Phi}$, we have 
$\calM,[\mybar{y} := \mybar{a}] \models \mu_j(\mybar{s})$.
The argumentation is analogous for clauses of the form (B).

Now let $\calM$ be a structure s.t.\ $\calM \models \exists \mybar{X}\,\psi$.
Let $R_1,\ldots,R_n$ be relations in $\calM$ s.t.\ $\calM \models \psi\unsubst{\mybar{X}}{\mybar{R}}$.
Then for all $j \in \{1,\ldots,n\}$:
\begin{align}
\calM \models & \forall \mybar{y}\, (\varphi \impl R_j(\mybar{s})),
 & \forall (\varphi,\mybar{s}) \in B_j,\\
\calM \models & \forall \mybar{y}\, (\varphi \land R_{i_1}(\mybar{t_1}) \land \cdots
 \land R_{i_m}(\mybar{t_m}) \impl R_j(\mybar{s})),
 & \forall (\varphi,\iota,\tau,\mybar{s}) \in I_j. \label{HornFixedPointProofInductionClause}
\end{align}
Assume $\mybar{a_j} \in F_{\Phi}(R_1, \ldots ,R_n)_j$.
Then there either exists $(\varphi,\mybar{s}) \in B_j$ s.t.\ 
 $\calM \models \exists \mybar{y}\, (\varphi \land \mybar{a_j} = \mybar{s})$
or there exists $(\varphi,\iota,\tau,\mybar{s}) \in I_j$ s.t.
$\calM,[\mybar{X}:= \mybar{R}] \models \exists \mybar{y}\, (\varphi \land \Land_{k = 1}^{m}
 X_{i_k}(\mybar{t_k}) \land \mybar{a_j} = \mybar{s})$.
We assume the latter, the proof for the former is analogous.
Thus let $\mybar{a}$ be s.t.
\[
\calM,[\mybar{X}:= \mybar{R}, \mybar{y} := \mybar{a}] \models \varphi \land \Land_{k = 1}^{m}
 X_{i_k}(\mybar{t_k}) \land \mybar{x_j} = \mybar{s}.
\]
From (\ref{HornFixedPointProofInductionClause}) we obtain $\calM,[\mybar{y} := \mybar{a}] \models
 R_j(\mybar{s})$ and thus $\calM \models R_j(\mybar{a_j})$.

Hence $F_{\Phi}(R_1, \ldots, R_n) \subseteq (R_1, \ldots ,R_n)$ and as $(\mu_1, \ldots, \mu_n)$
 is the least fixed point of $F_{\Phi}$ we obtain $\calM \models \Land_{j=1}^n ( \mu_j \impl R_j )$.

For all clauses in (E) we have
\[
\calM,[\mybar{X} := \mybar{R}] \models \forall \mybar{y}\,(\varphi \land X_{i_1}(\mybar{t_1}) \land \cdots
 \land X_{i_m}(\mybar{t_m}) \impl \bot),
\]
and therefore, as $X_1,\ldots,X_n$ occur only negatively in this formula, we obtain
\[
\calM,[\mybar{X} := \mybar{\mu}] \models \forall \mybar{y}\, (\varphi \land X_{i_1}(\mybar{t_1})
 \land \cdots \land X_{i_m}(\mybar{t_m}) \impl \bot).
\]
Thus $\calM,[\mybar{X} := \mybar{\mu}]$ satisfies all clauses in $\psi$ and we conclude that
 $\calM \models \psi \unsubst{\mybar{X}}{\mybar{\mu}}$.
For~\ref{lem.HornFP.minsol}.\ we get $\calM \models \Land_{j=1}^n ( \mu_j \impl R_j )$ analogously
 to the proof of~\ref{lem.HornFP.sol}.
\end{proof}

\begin{theorem}[Fixed-Point theorem]\label{thm.HornFP}
Let $\exists X_1 \cdots \exists X_n\, \psi$ be a valid Horn formula equation
 and let $\mu_j = [\lfp_{X_j} ~\Phi_{\psi}]$ for $j \in \{1,\ldots,n\}$. Then:
\begin{enumerate}
\item $\models \psi\unsubst{\mybar{X}}{\mybar{\mu}}$,
\item If $\models \psi \unsubst{\mybar{X}}{\mybar{\chi}}$ for FO[LFP]-formulas
 $\chi_1,\ldots,\chi_n$, then $\models  \bigwedge_{j=1}^n ( \mu_j \rightarrow \chi_j )$.
\end{enumerate}
\end{theorem}
\begin{proof}
Follows immediately from Lemma~\ref{lem.HornFP}.
\end{proof}

We now turn to dual and linear Horn formula equations.
A dual constrained Horn clause is an $L$-formula of the form
 $\varphi \land X(\mybar{t}) \impl Y_1(\mybar{s_1}) \lor \cdots \lor Y_n(\mybar{s_n})$ or
 $\varphi \impl  Y_1(\mybar{s_1}) \lor \cdots \lor Y_n(\mybar{s_n})$ where $\mybar{t},\mybar{s_1},
 \ldots,\mybar{s_n}$ are tuples of first-order terms of appropriate arity and $\varphi$ is a
 first-order formula, i.e., a formula not containing predicate variables.
A {\em dual Horn formula equation} is a formula equation of the form
 $\exists \mybar{X} \forall \mybar{y}\, \bigwedge_{i=1}^m H_i $, where $H_i$ is a constrained
 dual Horn clause for $i \in \{1,\ldots,m\}$. 
A {\em linear Horn formula equation} is a formula equation that is both Horn and dual Horn.

For a formula $\psi$ we define $\psi^D$ as $\psi\sop \sel{X_1}{\neg X_1}, \ldots, \sel{X_n}{\neg X_n} \scl$
 where $X_1,\ldots,X_n$ are all predicate variables occurring in $\psi$.
Note that $\models \psi \liff (\psi^D)^D$ for all formulas $\psi$.
Moreover, note that $\models \exists \mybar{X}\, \psi \liff \exists \mybar{X}\, \psi^D$ where $\mybar{X}=X_1,\ldots,X_n$
 are all predicate variables occurring in $\psi$.
If $\exists \mybar{X}\, \psi$ is a Horn formula equation, then $\exists \mybar{X}\, \psi^D$ is logically
 equivalent to a dual Horn formula equation and if $\exists \mybar{X}\, \varphi$ is a dual Horn formula
 equation, then $\exists \mybar{X}\, \varphi^D$ is logicall equivalent to a Horn formula equation.
Note that dualisation of a (dual) Horn formula equation interchanges (B)- and (E)-clauses.
\begin{example}
Consider the constrained Horn clauses
\[
\psi \quad\equiv\quad X(a) \land ( X(u)\land X(v) \impl Y(f(u,v)) ) \land (Y(w) \impl \bot).
\]
The dualisation of $\psi$ is
\[
\psi^D \quad\equiv\quad \neg X(a) \land (\neg X(u) \land \neg X(v) \impl \neg Y(f(u,v))) \land (\neg Y(w)\impl \bot)
\]
which is logically equivalent to the dual constrained Horn clauses
\[
(X(a) \impl \bot) \land (Y(f(u,v)) \impl X(u) \lor X(v)) \land Y(w).
\]
\end{example}
We can now prove the following result which is dual to Lemma~\ref{lem.HornFP}.
\begin{lemma}\label{lem.dualHornFP}
Let $\exists X_1 \cdots \exists X_n\, \psi$ be a dual Horn formula equation and
 $\nu_j := \neg [\lfp_{X_j} ~\Phi_{\psi^D}]$ for $j \in \{1,\ldots,n\}$, then
\begin{enumerate}
\item\label{lem.dualHornFP.sol} $\models \exists \mybar{X}\, \psi \leftrightarrow \psi \unsubst{\mybar{X}}{\mybar{\nu}}$
 and
\item\label{lem.dualHornFP.maxsol} if $\calM \models \psi \unsubst{\mybar{X}}{\mybar{R}}$,
 for a structure $\calM$ and relations $R_1, \ldots, R_n$ in $\calM$, then
 $\calM \models \Land_{j=1}^n ( R_j \rightarrow \nu_j )$.
\end{enumerate}
\end{lemma}
\begin{proof}
Let $\mu_j := [\lfp_{X_j} ~\Phi_{\psi^D}]$ for $j = 1,\ldots,n$.
For~\ref{lem.dualHornFP.sol}.\ note that, since $\exists \mybar{X}\, \psi$ is a dual Horn formula
 equation, $\exists \mybar{X}\, \psi^D$ is logically equivalent to a Horn formula equation. 
An application of Lemma~\ref{lem.HornFP}/\ref{lem.HornFP.sol}.\ yields
 $\models \exists \mybar{X}\, \psi^D \liff \psi^D \unsubst{\mybar{X}}{\mybar{\mu}}$.
Since $\psi^D\unsubst{\mybar{X}}{\mybar{\mu}}$ is syntactically equal to $\psi\unsubst{\mybar{X}}{\mybar{\nu}}$
 we obtain $\models \exists \mybar{X}\, \psi \liff \exists \mybar{X}\psi^D \liff \psi^D\unsubst{\mybar{X}}{\mybar{\mu}} \liff \psi\unsubst{\mybar{X}}{\mybar{\nu}}$.
 
For~\ref{lem.dualHornFP.maxsol}.\ assume that $\calM \models \psi\unsubst{\mybar{X}}{\mybar{R}}$ for a structure $\calM$
 and relations $R_1,\ldots,R_n$ in $\calM$.
Then $\calM \models \psi^D\sop \sel{X_1}{R_1\compl},\ldots,\sel{X_n}{R_n\compl} \scl$, so, 
by Lemma~\ref{lem.HornFP}/\ref{lem.HornFP.minsol}., $\calM \models \Land_{j=1}^n ( \mu_j \impl R_j\compl)$
 which yields $\calM \models \Land_{j=1}^n ( R_j \rightarrow \nu_j )$ by contraposition.
\end{proof}
\begin{theorem}[Dual Horn fixed-point theorem]
Let $\exists X_1 \cdots \exists X_n \psi$ be a valid dual Horn formula equation
 and let $\nu_j = \neg [\lfp_{X_j} ~\Phi_{\psi^D}]$ for $j = 1, \ldots ,n$,
 then
\begin{enumerate}
\item $\models \psi\unsubst{\mybar{X}}{\mybar{\nu}}$ and
\item if $\models \psi\unsubst{\mybar{X}}{\mybar{\chi}}$ for FO[LFP]-formulas
 $\chi_1, \ldots ,\chi_n$, then $\models \Land_{j=1}^n ( \chi_j \rightarrow \nu_j )$.
\end{enumerate}
\end{theorem}
\begin{proof}
Follows immediately from Lemma~\ref{lem.dualHornFP}.
\end{proof}

Note that the operator induced by $\Phi_{\psi^D}$ is not the dual operator of the one induced by $\Phi_\psi$
 in the sense of~\cite{Fritz01Fixed} because $\Phi_{\psi^D}$ is not the (pointwise) negation of
 $\Phi_\psi$.
Therefore $\nu$ is not the greatest fixed point of $\Phi_\psi$.
The question whether $\nu$ permits a sensible definition as a greatest fixed point is left as future work
 by this paper.
For the case of linear Horn formula equations we obtain:
\begin{theorem}[Linear Horn fixed-point theorem]\label{thm.linHorn}
Let $\exists X_1 \cdots \exists X_n\, \psi$ be a valid linear Horn formula equation,
 let $\mu_j = [\lfp_{X_j} ~\Phi_{\psi}]$ and $\nu_j = \neg [\lfp_{Y_j} ~\Phi_{\psi^D}]$ for
 $j = 1,\ldots,n$,
 then
\begin{enumerate}
\item $\models \psi \unsubst{\mybar{X}}{\mybar{\mu}}$ and
 $\models \psi \unsubst{\mybar{X}}{\mybar{\nu}}$ and
\item\label{thm.linHorn.interpol} if $\models \psi\unsubst{\mybar{X}}{\mybar{\chi}}$ for FO[LFP]-formulas
 $\chi_1,\ldots,\chi_n$, then
 $\models \Land_{j=1}^n ( ( \mu_j \rightarrow \chi_j) \land ( \chi_j \rightarrow \nu_j ))$.
\end{enumerate}
\end{theorem}
Theorem~\ref{thm.linHorn}/\ref{thm.linHorn.interpol}.\ shows that solving a linear Horn formula
 equation is equivalent to solving an interpolation problem in FO[LFP] in the sense that, given
 two tuples of FO[LFP]-formulas we seek to find a tuple of first-order, i.e., {\em fixed-point free}, 
 formulas which is between them in the implication ordering.
At least from a theoretical point of view this result helps to explain the efficacy of interpolation-based
 methods for solving constrained Horn clauses, see, e.g.,~\cite{McMillan12Solving}.
The relationship between interpolation and Horn clauses has also been studied by
 encoding interpolation problems with a language condition on the constant symbols as Horn clause
 sets~\cite{Ruemmer13Classifying,Gupta14Generalised}.

\section{Applications to program verification}~\label{sec.apppv}

In this section we will describe some direct applications of our fixed-point theorem to
 the foundations of program verification.
As an exemplary framework we will consider the Hoare calculus for a simple imperative programming
 language as in~\cite{Winskel93Formal}.
We fix the first-order language of arithmetic $\calL = \{0,1,+,-,\cdot,\leq\}$.
The set of programs is defined by
\begin{align*}
p ::= \textbf{skip} ~|~ x := t ~|~ p_0;p_1 ~|~ \textbf{if}~ B~\textbf{then}~ p_0 ~\textbf{else} ~p_1 ~|~ \textbf{while}~ B ~\textbf{do}~ p_0,
\end{align*}
where $t$ is an $\calL$-term, $B$ a quantifier-free first-order formula in $\calL$
 and $x$ is a program variable.

The denotational semantics of programs is defined as usual based on a set of states $\Sigma$:
we write $C(p)$ for the partial function from $\Sigma$ to $\Sigma$ that is the denotational
 semantics of the program $p$.
A Hoare triple is written as $\hoare{\varphi}{p}{\psi}$.
For the purposes of this paper we fix the program variables to taking values in the integers and thus
 we can work in the standard model $\Int$.
We write $\sigma \models \hoare{\varphi}{p}{\psi}$ if $\Int,\sigma\models \varphi$ implies
 that $\Int,C(p)(\sigma) \models \psi$ and $\models \hoare{\varphi}{p}{\psi}$
 if $\sigma \models \hoare{\varphi}{p}{\psi}$ for all $\sigma \in \Sigma$.
The Hoare calculus can be defined as usual, see, e.g.~\cite{Winskel93Formal}.
We write $\proves \hoare{\varphi}{p}{\psi}$ if $\hoare{\varphi}{p}{\psi}$ is provable
 in the Hoare calculus.
We can then consider the verification condition of a Hoare triple as a Horn formula
 equation as follows:
\begin{definition}
The \emph{verification condition} of a Hoare triple $\hoare{\varphi}{p}{\psi}$,
 written $\vc(\hoare{\varphi}{p}{\psi})$, is a formula equation
 $\exists \mybar{I} \forall^*\, \vct(\hoare{\varphi}{p}{\psi})$,
 where $\vct(\hoare{\varphi}{p}{\psi})$ is defined by structural induction on $p$ as follows:
\begin{align*}
\vct(\hoare{\varphi}{\hskipit}{\psi}) &= \varphi \impl \psi\tag{1}\\
\vct(\hoare{\varphi}{x_j := t}{\psi}) &= \varphi \impl \psi\unsubst{x_j}{t}\tag{2}\\
\vct(\hoare{\varphi}{p_0;p_1}{\psi}) &= \vct(\hoare{\varphi}{p_0}{I}) \land \vct(\hoare{I}{p_1}{\psi})\tag{3}\\
\vct(\hoare{\varphi}{\hif B \hthen p_0 \helse p_1}{\psi}) &= \vct(\hoare{\varphi\land B}{p_0}{\psi}) \land \vct(\hoare{\varphi\land \neg B}{p_1}{\psi})\tag{4}\\
\vct(\hoare{\varphi}{\hwhile B \hdo p_0}{\psi}) &= \vct(\hoare{I \land B}{p_0}{I}) \land (\varphi \impl I) \land (I \land \neg B \rightarrow \psi)\tag{5}
\end{align*}
where, in clauses (3) and (5), $I$ is a fresh predicate variable which does not appear in
 $\varphi$ nor in  $\psi$.
Then $\vc(\hoare{\varphi}{p}{\psi}) = \exists \mybar{I} \forall^*\, \vct(\hoare{\varphi}{p}{\psi})$ is
 obtained from $\vct(\hoare{\varphi}{p}{\psi})$ by
 existential quantification of all predicate variables in $\forall^*\, \vct(\hoare{\varphi}{p}{\psi})$.
\end{definition}
Note that this is an purely syntactic definition, thus we can define $\vc(\hoare{\varphi}{p}{\psi})$
 analogously for a program $p$ and second-order formulas $\varphi, \psi$.
We then obtain the following completeness result which characterises Hoare provability by truth of a formula
 equation.
\begin{theorem}\label{thm.HoareSemantics}
Let $\hoare{\varphi}{p}{\psi}$ be a Hoare triple.
Then $\proves \hoare{\varphi}{p}{\psi}$ iff $\Int \models \hoare{\varphi}{p}{\psi}$ iff $\Int \models \vc(\hoare{\varphi}{p}{\psi})$.
\end{theorem}
\begin{proof}[Proof Sketch]
$\proves \hoare{\varphi}{p}{\psi}$ iff $\Int \models \hoare{\varphi}{p}{\psi}$ is soundness and completeness
 of the Hoare calculus.
The implication from $\proves \hoare{\varphi}{p}{\psi}$ to $\Rightarrow \quad \Int \models \vc(\hoare{\varphi}{p}{\psi})$
 is proved by a straightforward induction on the structure of the Hoare proof of $\hoare{\varphi}{p}{\psi}$
The implication from $\Int \models \vc(\hoare{\varphi}{p}{\psi})$ to $\Int \models \hoare{\varphi}{p}{\psi}$ is proved
 by translating the semantics.
\end{proof}
As a corollary we can now obtain the statement that a partial correctness assertion of an imperative program
 is expressible as a formula in FO[LFP], a point also made in~\cite{Blass87Existential} for existential least
 fixed-point logic.
\begin{corollary}
Let $\hoare{\varphi}{p}{\psi}$ be a Hoare triple and let $\mybar{\mu}$ be the solution of
 $\vc(\hoare{\varphi}{p}{\psi})$ obtained from Theorem~\ref{thm.HornFP}, then
$\Int \models \vc(\hoare{\varphi}{p}{\psi})$ iff
 $\Int \models \vct(\hoare{\varphi}{p}{\psi})\unsubst{\mybar{I}}{\mybar{\mu}}$.
\end{corollary}
A second corollary is based on the observation that
 $\vc(\hoare{\varphi}{p}{\psi})$ is a linear Horn formula equation (which can be shown
 by a straightforward induction).
Therefore we can apply Theorem~\ref{thm.linHorn} to the verification condition and obtain:
\begin{corollary}\label{cor.hoareipol}
Let $\hoare{\varphi}{p}{\psi}$ be a Hoare triple,
 then $\exists I_1 \cdots \exists I_n\, \pi \equiv \vc(\hoare{\varphi}{p}{\psi})$ 
 is a linear Horn formula equation.
Let $\pi^D$ be the dual formula of $\pi$ with predicate variables $K_1,\ldots,K_n$.
Assume $\Int \models \vc(\hoare{\varphi}{p}{\psi})$ and let
 $\mu_j = [\lfp_{I_j} ~\Phi_{\pi}]$ and $\nu_j = \neg [\lfp_{K_j} ~\Phi_{\pi^D}]$
 for $j = 1, \ldots, n$.
Then 
\begin{enumerate}
\item $\Int \models \vct(\hoare{\varphi}{p}{\psi})\unsubst{\mybar{I}}{\mybar{\mu}}$
 and $\Int \models \vct(\hoare{\varphi}{p}{\psi})\unsubst{\mybar{I}}{\mybar{\nu}}$.
\item If $\Int,[\mybar{I} := \mybar{R}] \models \vct(\hoare{\varphi}{p}{\psi})$
 for relations $R_1, \ldots ,R_n$, then $\Int \models \Land_{j=1}^n ( ( \mu_j \impl R_j)
  \land (R_j \impl \nu_j))$.
\end{enumerate}
\end{corollary}
This corollary shows that finding first-order formulas as loop invariants is equivalent to
 an interpolation problem in the sense of finding a fixed-point free interpolant.
Just as Theorem~\ref{thm.linHorn} does for linear Horn clauses, this corollary contributes
 to explaining the efficacy of interpolation-based methods for loop invariant generation,
 see, e.g.~\cite{McMillan03Interpolation}.

As a third corollary we will show that the weakest precondition and the strongest postcondition
 are expressed by the least and greatest solutions $\mu$ and $\nu$ of linear Horn formula equations
 based on the verification condition.
\begin{definition}\label{hoareWeakestPreconditionDefinition}
Let $p$ be a program and $\varphi, \psi$ be first-order formulas in $\calL$.
The \emph{weakest precondition}\footnote{In the literature this is mostly called weakest liberal
 precondition and the term weakest precondition is reserved for the context of total correctness.
 As we only talk about partial correctness of programs there is no need for us to do so.}
of $p$ and $\psi$, written $\wpc(p,\psi)$, is defined as
\[
\wpc(p,\psi) = \{\sigma \in \Sigma ~|~ \Int,C(p)(\sigma) \models \psi \}.
\]
The \emph{strongest postcondition} of $p$ and $\varphi$, written $\spc(p,\varphi)$, is defined as
\[
\spc(p,\varphi) = \{\sigma \in \Sigma ~|~ \exists \sigma' \in \Sigma: \Int,\sigma' \models \varphi ~\text{and}~ C(p)(\sigma') = \sigma \}.
\]
\end{definition}
The following property of the weakest precondition and the strongest postcondition
 justifies the terminology and is of fundamental importance.
\begin{lemma}\label{lem.spwp}
Let $\hoare{\varphi}{p}{\psi}$ be a Hoare triple, then $\models \hoare{\varphi}{p}{\psi}$
 iff $[\varphi] \subseteq \wpc(p,\psi)$ iff $[\psi] \supseteq \spc(p,\varphi)$.
\end{lemma}
For a formula $\varphi$ all of whose free variables are program variables we define 
 $[\varphi] = \{ \sigma\in\Sigma \mid \Int,\sigma\models \varphi\}$,
 the set of states defined by $\varphi$.
It is well-known that for any program $p$ and any formula $\psi$ there is a first-order formula
 $\varphi_\wpc$ which defines $\wpc(p,\psi)$, i.e., $[\varphi_\wpc] = \wpc(p,\psi)$ and, symmetrically,
 for any program $p$ and any formula $\varphi$ there is a first-order formula $\psi_\spc$ which
 defines $\spc(p,\varphi)$, i.e., $[\psi_\spc] = \spc(p,\varphi)$.
Note that these formulas rely on the expressivity of the assertion language, i.e, in our setting, 
 on an encoding of finite sequences in $\Int$.

We will consider the formula equation $\exists X_0\, \vc(\hoare{\varphi}{p}{X_0})$, which asks for
 a formula $X_0$ s.t.\ all states satisfying $\varphi$ satisfy $X_0$ after running the program $p$.
Symmetrically we will consider the formula equation $\exists X_0\, \vc(\hoare{X_0}{p}{\psi})$. 
Note that these are linear Horn formula equations and therefore we can apply Theorem~\ref{thm.linHorn}.
In general there also occur predicate variables in $\vc(\hoare{\varphi}{p}{\psi})$, yet here we will only
 be interested in the solution for the outermost predicate variable.
\begin{corollary}\label{cor.spc}
Let $p$ be a program, let $\varphi$ be a formula,
 let $\exists \mybar{X}\, \pi \equiv \exists X_0\, \vc(\hoare{\varphi}{p}{X_0})$,
 and let $\mu = [\lfp_{X_0}\, \Phi_{\pi}]$, then $[\mu] = \spc(p,\varphi)$.
\end{corollary}
\begin{proof}
Since $\Int \models \vc(\hoare{\varphi}{p}{\top})$ we have
 $\Int \models \exists X\, \vc(\hoare{\varphi}{p}{X_0})$.
From Lemma~\ref{lem.HornFP}/\ref{lem.HornFP.sol}.\ we obtain $\Int \models \vc(\hoare{\varphi}{p}{\mu})$
 which, by Theorem~\ref{thm.HoareSemantics}, is equivalent to $\models \hoare{\varphi}{p}{\mu}$.
By Lemma~\ref{lem.spwp} we obtain $[\mu] \supseteq \spc(p,\varphi)$.

For the other direction let $\psi_{\spc}$ be a first-order formula with $[\psi_{\spc}] = \spc(p,\varphi)$.
By Lemma~\ref{lem.spwp} we have $\models \hoare{\varphi}{p}{\psi_{\spc}}$, which is equivalent to
 $\Int \models \vc(\hoare{\varphi}{p}{\psi_{\spc}})$.
By Lemma~\ref{lem.HornFP}/\ref{lem.HornFP.minsol}.\ we have $\Int \models \forall \mybar{x}
 ( \mu(\mybar{x}) \impl \psi_{\spc}(\mybar{x}))$ and thus
 $[\mu] \subseteq [\psi_{\spc}] = \spc(p,\varphi)$.
\end{proof}

\begin{corollary}
Let $p$ be a program and let $\psi$ be a formula,
 let $\exists \mybar{X}\, \pi \equiv \exists X_0\, \vc(\hoare{X_0}{p}{\psi})$,
 and let $\nu = \neg [\lfp_{X_0}\,\Phi_{\pi^D}]$, then $[\nu] = \wpc(p,\psi)$.
\end{corollary}
\begin{proof}
Symmetric to that of Corollary~\ref{cor.spc}
\end{proof}

Note that the formulas $\mu$ and $\nu$ thus obtained do not rely on an expressivity hypothesis anymore.
The encoding of sequences is replaced by the least fixed-point operator.

\section{Towards an abstract fixed-point theorem}\label{sec.afpthm}

Abstract interpretation, originally introduced in~\cite{Cousot77Abstract}, is one of the most important techniques
 in static analysis and software verification.
Since many verification techniques could successfully be generalised from programs to the logical level
 of constrained Horn clauses it is also natural to expect this possibility for abstract interpretation.
And indeed, abstract interpretations have been used in tools for solving Horn clauses~\cite{Hoder11muZ,Kafle16Rahft}.
In this section we briefly outline how we expect this generalisation to apply to our fixed-point theorem.
Proving the main statement of this section, Conjecture~\ref{conj.afpthm}, is currently work-in-progress.

An application we have in mind is the following: in the recent article~\cite{Hetzl20Decidability}
 the decidability of the existence of affine invariants for programs with affine assignments
 (essentially due to Karr~\cite{Karr76Affine}) has been generalised
 to formula equations of the form $\exists \mybar{X} \forall^* \varphi$ for $\varphi$ being a quantifier-free
 formula.
The essential difference between the proof in~\cite{Hetzl20Decidability} and 
 Theorem~\ref{thm.HornFP} is that in ~\cite{Hetzl20Decidability} the fixed point is formed in the lattice of affine subspaces and not in
 the power set lattice.
We can allow for this possibility as follows.
\begin{definition}
Let $\calA = (A,\subseteq)$ and $\calB = (B,\sqsubseteq)$ be two partially ordered sets.
A \emph{Galois connection} between $\calA$ and $\calB$ consists of two functions
$\alpha: A \rightarrow B$ and $\gamma: B \rightarrow A$, s.t.\ for all $X \in A$
and $Y \in B$:
\begin{align*}
X \subseteq \gamma(Y) \quad \Leftrightarrow \quad \alpha(X) \sqsubseteq Y.
\end{align*}
\end{definition}

\begin{definition}
Let $\calL$ be a language.
A \emph{model abstraction} is a pair $(\calM,G)$, where $\calM$ is an $\calL$-structure and
 $G = (\mathcal{V}_k,\alpha_k,\gamma_k)_{k \in \Nat}$ is a sequence of triples, s.t.\ for all
 $k \in \Nat:$ $\mathcal{V}_k =(V_k,\sqsubseteq)$ is a complete lattice and
 $\alpha_k: \powset(M^k) \rightarrow V_k$ and $\gamma_k: V_k \rightarrow \powset(M^k)$ form a Galois-connection
 between $(\powset(M^k),\subseteq)$ and $\mathcal{V}_k$.
\end{definition}

\begin{example}\label{ex.affone}
Consider the language $\Laff = (0,1,+,(c)_{c \in \Rat})$ where the intended interpretation of the
 unary function symbol $c$ for $c\in \Rat$ is multiplication with $c$.
Let $G_{\aff} = ((\Aff \Rat^k,\subseteq),\aff_k,\id_k)_{k \in \Nat}$, where $\Aff \Rat^k$ is the set of all affine
subspaces of $\Rat^k$, $\aff_k$ maps every subset of $\Rat^k$ to its affine hull and $\id_k$ is the
 embedding of $\Aff \Rat^k$ in $\powset(\Rat^k)$.
Then $(\Rat, G_{\aff})$ is a model abstraction. 
\end{example}

We will now introduce the semantics of model abstractions.
We will interpret FO[LFP]- and second-order-formulas $\varphi$ by defining a satisfaction relation
 $(\calM, G) \modelsa \varphi$.
The crucial difference between $\modelsa$ and standard Tarski semantics $\models$ will be that
 second-order quantifiers and the least fixed-point operator will not be interpreted in the power set
 of the domain but in $G$ instead (for the appropriate arity).
This restriction of the domain of second-order quantification is reminiscent of (but different from)
 Henkin semantics of second- and higher-order logic~\cite{Henkin50Completeness}.
\begin{definition}
The defining clauses for first-order atoms, propositional connectives, and first-order
 quantifiers for $\modelsa$ are identical to those for $\models$.
For formulas of the form $\exists X\, \psi$, where $X$ is a $k$-ary predicate variable, we define
\[
(\calM,G) \models_a \exists X \psi \quad \Leftrightarrow \quad \exists S \in V_k: (\calM,G) \models_a \psi[X\backslash \gamma_k(S)],
\]
and analogously for formulas of the form $\forall X \psi$.
The semantics of the $\lfp$-operator is defined as follows.
Let $X_1,\ldots,X_n$ be formula variables with $X_i$ having arity $k_i$, let
 $\Phi = (\varphi_i(X_1,\ldots,X_n,\mybar{u_i}))_{i=1}^n$ be a tuple of formulas
 s.t.\ $|\mybar{u_i}| = k_i$ and $X_1,\ldots,X_n$ occur only positively in $\varphi_i$
 for all $i\in\{1,\ldots,n\}$.
Let
\begin{align*}
F_i^{\#}: ~V_{k_1} \times \cdots \times V_{k_n} &\rightarrow V_{k_i}\\
(Y_1, \ldots ,Y_n) &\mapsto \alpha_{k_i} \circ F_i(\gamma_{k_1}(Y_1), \ldots ,\gamma_{k_n}(Y_n)),
\end{align*}
where $F_i$ is defined as in Section~\ref{sec.fpthm} and let $F_{\Phi}^{\#} = (F_1^{\#},\ldots,F_n^{\#})$.
Then
\[
(\calM,G) \modelsa [\lfp_{X_i} ~\Phi](\mybar{a}) \quad \Leftrightarrow \quad \mybar{a} \in \gamma_{k_i}( \lfp(F_{\Phi}^{\#})_i).
\]
\end{definition}
Note that the semantics of the least fixed-point operator is well-defined:
We already know that $F_i$ is a monotone operator and $\alpha_{k_i}$ and $\gamma_{k_i}$ are monotone
 as they form a Galois connection for $i \in \{1,\ldots,n\}$.
Thus $F_i^{\#}$ is monotone for all $i \in \{1,\ldots,n\}$ and therefore $F_{\Phi}^{\#}$ is monotone as
well. As $\mathcal{V}_k$ is a complete lattice for all $k \in \Nat$, we can use the Knaster-Tarski theorem
 to obtain the least fixed point of $F_{\Phi}^{\#}$.

\begin{conjecture}[Abstract fixed-point theorem]\label{conj.afpthm}
Let $\calL$ be a language and $(\calM,G)$ be a model abstraction.
Let $\exists X_1 \cdots \exists X_n\, \psi$ be a Horn
 formula equation and let $\mu_j = [\lfp_{X_j} ~\Phi_{\psi}]$ for $j=1,\ldots,n$.
Then:
\begin{enumerate}
\item $(\calM,G) \models_a \exists \mybar{X}\, \psi \leftrightarrow \psi\unsubst{\mybar{X}}{\mybar{\mu}}$ and
\item if $(\calM,G) \models_a \psi\unsubst{\mybar{X}}{\mybar{\chi}}$ for
 FO[LFP]-formulas $\chi_1,\ldots,\chi_n$, then $(\calM,G) \models_a \bigwedge_{j=1}^n \left( \alpha_j \rightarrow \chi_j \right)$.
\end{enumerate}
\end{conjecture}

\begin{example}
In continuation of Example~\ref{ex.affone} consider an $\Laff$ formula equation $\exists \mybar{X}\, 
 \forall^* \psi$ where $\psi$ is a quantifier-free formula.
We want to decide whether there are formulas $\mybar{\chi}$ s.t.\ $\Rat \models \psi[\mybar{X}\backslash \mybar{\chi}]$
 and $\chi_1,\ldots,\chi_n$ are conjunctions of affine equations, i.e., they define affine subspaces of
 $\Rat^k$ in the model $\Rat$.
As in~\cite{Hetzl20Decidability} we reduce the solvability of $\exists \mybar{X} \forall^*\, \psi$ to
 solvability of one of its finitely many projections which are Horn formula equations.
Let $\exists \mybar{X} \forall^*\, \varphi$ be one of them, then an application
 of Conjecture~\ref{conj.afpthm} yields a tuple of FO[LFP]-formulas $\mybar{\mu}$ s.t.\ 
 $(\Rat,G_{\aff}) \models_a \exists \mybar{X}\, \varphi \leftrightarrow \varphi\unsubst{\mybar{X}}{\mybar{\mu}}$.
Since all lattices $\Aff \Rat^k$ have finite height we can compute fixed-point free formulas $\mybar{\chi}$
 equivalent to $\mybar{\mu}$ from $\mybar{\mu}$ and therefore $(\Rat,G_{\aff}) \models_a \exists \mybar{X}\, \varphi \leftrightarrow \varphi\unsubst{\mybar{X}}{\mybar{\chi}}$.
Now $\varphi\unsubst{\mybar{X}}{\mybar{\chi}}$ is a first-order formula and hence $(\Rat,G_{\aff}) \models_a \varphi\unsubst{\mybar{X}}{\mybar{\chi}}$
 iff $\Rat \models \varphi\unsubst{\mybar{X}}{\mybar{\chi}}$.
The latter statement can now be checked by a decision procedure for linear arithmetic.
\end{example}

\section{Conclusion}

We have shown a fixed-point theorem for Horn formula equations and applied it to dual Horn formula equations
 and linear Horn formula equations.
The proof of this result essentially consists of expressing the construction of a minimal model
 of a set of Horn clauses, which is familiar from logic programming, on the object level as a formula
 in first-order logic with least fixed points, thus providing
 a canonical solution to a Horn formula equation in FO[LFP].

Note that Theorem~\ref{thm.HornFP} applies equally to constraints $\varphi$ being FO[LFP]-formulas.
It therefore shows that FO[LFP], in contrast to first-order logic, has the property of being
 closed under solving Horn formula equations.
It thus shows that in FO[LFP] validity and solvability of Horn formula equations coincide.
This is in contrast to formula equations in first-order logic, cf.~Example~\ref{ex.feq}.

Explicit fixed-point operators have been employed for second-order quantifier elimination
 in the $\mathrm{DLS}^*$ algorithm introduced in~\cite{Nonnengart98Fixpoint,Doherty98General}.
In this algorithm Ackermann's lemma is generalised to a fixed-point lemma that covers additional situations.
While the $\mathrm{DLS}^*$ algorithm as such will fail on Horn formula equations due to its priority
 of shifting universal quantifiers inwards, the fixed-point lemma of the $\mathrm{DLS}^*$ algorithm
 could be used for solving Horn formula equations with a single variable when combined with
 a different preprocessing.
However, for Horn formula equations with more than one predicate variable it would result
 in solutions with iterated fixed points.
In this sense our fixed-point theorem can be considered a generalisation of the fixed-point lemma
 of~\cite{Nonnengart98Fixpoint}.
On the other hand the fixed-point lemma of~\cite{Nonnengart98Fixpoint} is not restricted to
 Horn formula equations.

These fixed-point theorems contribute to our theoretical understanding of the logical foundations of
 constrained Horn clause solving and software verification: 
Theorem~\ref{thm.linHorn} and Corollary~\ref{cor.hoareipol} help to explain, from a theoretical
 point of view, the efficacy of interpolation for Horn clause solving and invariant generation respectively.
Moreover, as corollary to our fixed point theorem we have obtained the expressibility of the 
 weakest precondition and the strongest postcondition, and thus the partial correctness of an imperative
 program in FO[LFP].

As shown in~\cite{Kloibhofer20Fixed}, our fixed-point theorem has a number of further
 corollaries in a variety of application areas: it allows a generalisation of a result by
 Ackermann~\cite{Ackermann35Untersuchungen} on second-order quantifier-elimination in a direction different
 from the recent generalisation~\cite{Wernhard17Approximating} of that result.
It allows to obtain a result on the generation of a proof with induction based on partial information
 about that proof shown in~\cite{Eberhard15Inductive} as straightforward corollary.
Last, but not least, an abstract form of the fixed-point theorem, Conjecture~\ref{conj.afpthm}, would
 permit to considerably simplify the proof of the decidability of affine formula
 equations given in~\cite{Hetzl20Decidability}.
 
In conclusion, we believe that it is fruitful to consider constrained Horn clause solving from the more general
 point of view of solving formula equations.
On the theoretical level this perspective uncovers connections to a number of topics such as
 second-order quantifier elimination and results such as Ackermann's~\cite{Ackermann35Untersuchungen}.
On the practical level it suggests to study the applicability of algorithms such as DLS and SCAN for
 constrained Horn clauses and vice versa, that of algorithms for constrained Horn clause solving for
 applications of second-order quantifier elimination.

{\bf Acknowledgements.} The authors would like to thank Christoph Wernhard for a number of helpful
 conversations about formula equations and second-order quantifier elimination and the anonymous
 reviewers for many useful suggestions that have improved our work on this subject.
 
\bibliographystyle{eptcs}
\bibliography{references}

\end{document}